\documentclass[11pt,letterpaper]{article}

\usepackage{amsmath,amsfonts,amssymb,amsthm}
\usepackage{amsthm}
\usepackage{graphicx,color}
\usepackage{boxedminipage}
\usepackage[ruled,boxed,linesnumbered]{algorithm2e}

\SetCommentSty{mycommfont}
\usepackage{framed}
\usepackage{thmtools}
\usepackage{thm-restate}
\usepackage{xspace}
\usepackage{todonotes}
\usepackage{footnote}
\usepackage{xcolor}

\usepackage[margin=1in,nohead]{geometry}
\usepackage{todonotes}
\usepackage{footnote}
 \usepackage[pdftex, plainpages = false, pdfpagelabels, 
                 bookmarks=false,
                 bookmarksopen = true,
                 bookmarksnumbered = true,
                 breaklinks = true,
                 linktocpage,
                 pagebackref,
                 colorlinks = true,  
                 linkcolor = blue,
                 urlcolor  = blue,
                 citecolor = red,
                 anchorcolor = green,
                 hyperindex = true,
                 hyperfigures
                 ]{hyperref} 
 \usepackage[nameinlink]{cleveref}
 \usepackage{xifthen}
 \usepackage{tabularx}
\usepackage{tikz} 
 \usetikzlibrary{calc,decorations.pathreplacing}

\newtheorem{lemma}{Lemma}

\newtheorem{corollary}{Corollary}
\newtheorem{definition}{Definition}

\theoremstyle{definition}

\newcommand{\opt}{{\sf OPT}}

\newcommand{\cT}{\mathcal{T}}
\newcommand{\cP}{\mathcal{P}}

\newlength{\RoundedBoxWidth}
\newsavebox{\GrayRoundedBox}
\newenvironment{GrayBox}[1]%
   {\setlength{\RoundedBoxWidth}{.93\textwidth}
    \def\boxheading{#1}
    \begin{lrbox}{\GrayRoundedBox}
       \begin{minipage}{\RoundedBoxWidth}}%
   {   \end{minipage}
    \end{lrbox}
    \begin{center}
    \begin{tikzpicture}%
       \node(Text)[draw=black!20,fill=white,rounded corners,%
             inner sep=2ex,text width=\RoundedBoxWidth]%
             {\usebox{\GrayRoundedBox}};
        \coordinate(x) at (current bounding box.north west);
        \node [draw=white,rectangle,inner sep=3pt,anchor=north west,fill=white] 
        at ($(x)+(6pt,.75em)$) {\boxheading};
    \end{tikzpicture}
    \end{center}}

\newenvironment{defproblemx}[2][]{\noindent\ignorespaces%
                                \FrameSep=6pt%
                                \parindent=0pt%
                \vspace*{-1.5em}
                \ifthenelse{\isempty{#1}}{%
                  \begin{GrayBox}{\textsc{#2}}%
                }{%
                  \begin{GrayBox}{\textsc{#2} parameterized by~{#1}}%
                }
                \begin{tabular*}{\textwidth}{@{\hspace{.1em}} >{\itshape} p{1.8cm} p{0.8\textwidth} @{}}%
            }{
                \end{tabular*}%
                \end{GrayBox}%
                \ignorespacesafterend
            }

\newcommand{\Oh}{\mathcal{O}}

\newcommand{\Ggp}{$G(k,p)$\xspace}

\begin{document}

\title{Shortest Cycles With Monotone Submodular Costs\thanks{The results of this paper will appear in the 
\emph{Proceedings of the 34th Annual ACM-SIAM Symposium on Discrete Algorithms (SODA 2023)}.
The research leading to these results has received funding from the Research Council of Norway via the project BWCA (grant no. 314528). Giannos Stamoulis acknowledges support by the ANR project ESIGMA (ANR-17-CE23-0010) and the French-German Collaboration ANR/DFG Project UTMA (ANR-20-CE92-0027).}
}

\author{
Fedor V. Fomin\thanks{
Department of Informatics, University of Bergen, Norway. Emails:
\texttt{fomin@ii.uib.no},
\texttt{petr.golovach@uib.no},
\texttt{tuukka.korhonen@uib.no}
}
\and
Petr A. Golovach\addtocounter{footnote}{-1}\footnotemark{}
\and
Tuukka Korhonen\addtocounter{footnote}{-1}\footnotemark{}
\and 
Daniel Lokshtanov\thanks{Department of Computer Science, University of California, Santa Barbara, USA.
Email: \texttt{daniello@ucsb.edu}}
\and 
Giannos Stamoulis\thanks{LIRMM, Univ Montpellier, CNRS, Montpellier, France. Email: \texttt{giannos.stamoulis@lirmm.fr}}
}

\date{}

\pagenumbering{Alph}
\maketitle
\thispagestyle{empty}

\begin{abstract}
We introduce the following submodular generalization of the \textsc{Shortest Cycle} problem. 
For a nonnegative monotone submodular cost function $f$  defined on the edges (or the vertices) of an undirected graph $G$, we seek for a cycle $C$ in $G$ of minimum cost $\opt=f(C)$. 
We give an algorithm that given an $n$-vertex graph $G$, parameter $\varepsilon > 0$, and the function $f$ represented by an oracle, in time  $n^{\Oh(\log 1/\varepsilon)}$ finds a cycle $C$ in $G$ with $f(C)\leq (1+\varepsilon)\cdot \opt$.
This is in sharp contrast with the non-approximability of the closely related \textsc{Monotone Submodular Shortest $(s,t)$-Path} problem, which 
requires exponentially many queries to the oracle
for finding an $n^{2/3-\varepsilon}$-approximation [Goel et al., FOCS 2009].
We complement our algorithm with a matching lower bound.
We show that for every $\varepsilon > 0$, obtaining a $(1+\varepsilon)$-approximation requires at least $n^{\Omega(\log 1/ \varepsilon)}$ queries to the oracle.

When the function $f$ is integer-valued,  our algorithm yields that a cycle of cost $\opt$ can be found in time $n^{\Oh(\log \opt)}$. In particular, for  $\opt=n^{\Oh(1)}$  this gives a quasipolynomial-time algorithm computing a cycle of minimum submodular cost.
Interestingly, while a quasipolynomial-time algorithm often serves as a good indication that a polynomial time complexity could be achieved, we show a lower bound that $n^{\Oh(\log n)}$ queries are required even when $\opt = \Oh(n)$. 
\end{abstract} 

\newpage
\pagenumbering{arabic}
\pagestyle{plain}
\setcounter{page}{1}

\section{Introduction}\label{sec:intro}

Submodular function minimization is a fundamental problem in combinatorial optimization.  
This problem is solvable in (strongly) polynomial time \cite{cunningham1985submodular,grotschel1981ellipsoid,iwata2001combinatorial,iwata2009simple,schrijver2000combinatorial}.
However, the problem becomes intractable even with straightforward additional cardinality constraints 
 \cite{GoemansHIM09,svitkina2011submodular}.
A significant amount of research on submodular optimization is on generalizing  the classical computer science problems by 
replacing simpler objective functions with general submodular functions.  Examples of submodular minimizations over combinatorial constraints include load balancing, balanced cut \cite{svitkina2011submodular}, vertex cover \cite{goel2009approximability,iwata2009submodular,wolsey1982analysis}, shortest path, perfect matching, spanning tree \cite{goel2009approximability} or min-cut 
\cite{jegelka2009notes}.

However, it seems that for almost every natural graph problem in P (shortest $(s,t)$-path, matching, spanning tree, or minimum $(s,t)$-cut) its submodular generalizations becomes hard. 
Let  $f\colon 2^{E(G)}\rightarrow \mathbb{R}_{\geq 0}$  be a monotone submodular cost function defined by a value-giving oracle on the edges of an undirected graph $G$ with  $m$ edges and $n$ vertices.  The following computational tasks require exponentially many queries to the value oracle:
\begin{itemize}
\item Finding an   $\Oh(n^{2/3-\varepsilon})$-approximation of the minimum cost of an $(s,t)$-path (\textsc{Submodular Shortest $(s,t)$-Path})  
 \cite{goel2009approximability}; 
\item Finding an   $\Oh(n^{1-\varepsilon})$-approximation of the minimum cost of a  perfect matching (\textsc{Submodular Perfect Matching})  \cite{goel2009approximability}; 
\item Finding an   $\Oh(n^{1-\varepsilon})$-approximation of the minimum cost of a  spanning tree (\textsc{Submodular Minimum Spanning Tree})  \cite{goel2009approximability}; 
\item  Finding an   $\Oh(n^{1/3-\varepsilon})$-approximation of the minimum cost of an $(s,t)$-cut (\textsc{Submodular Minimum  $(s,t)$-Cut}) \cite{jegelka2009notes}. 
\end{itemize}

We discover an interesting anomaly, a classical problem in P, whose monotone submodular generalization strongly deviates from this common pattern. This is the problem of computing the girth, that is,  the length of a shortest cycle, of an undirected graph.
In sharp contrast to all these non-approximability results,  we show that the problem of finding a cycle in a graph with minimum monotone submodular cost admits a polynomial-time approximation scheme (PTAS) and a quasipolynomial-time algorithm when the values of the submodular function are polynomially-bounded integers.
More precisely,  for a graph $G$ and a function \mbox{$f\colon 2^{V(G)}\rightarrow \mathbb{R}_{\geq 0}$}, we define 
\mbox{$\opt=\min\{f(C)\colon C\subseteq V(G) \text{ induces a cycle of  }G\}$}.
Our first main result is the following theorem.

\begin{restatable}{theorem}{thmmaincycle}
\label{thm:min-cycle}
There is an algorithm that given an $n$-vertex graph $G$, parameter $\varepsilon > 0$, and a monotone submodular function $f\colon 2^{V(G)}\rightarrow \mathbb{R}_{\geq 0}$ represented by an oracle, finds a cycle $C$ in $G$ with $f(C)\leq (1+\varepsilon)\cdot \opt$ in time $n^{\Oh(\log 1/\varepsilon)}$.
\end{restatable}

We stated \Cref{thm:min-cycle} for a function $f$ defined on the vertices of a graph.
An easy reduction by placing a new vertex on every edge shows that the same result holds for monotone submodular functions defined on the edges of a multigraph, see \Cref{cor:min-cycle}.



When the function $f$ is integer-valued,  \Cref{thm:min-cycle} (by setting  $\varepsilon = \frac{1}{w +1}$ with $\opt\leq w\leq 2\opt$, where $w=f(C)$ for the cycle $C$ returned by the approximation algorithm for $\varepsilon=1/2$)    
implies that a cycle of cost $\opt$ can be found in time $n^{\Oh(\log \opt)}$. In particular, when $\opt=n^{\Oh(1)}$, it gives a quasipolynomial-time algorithm computing a cycle of minimum monotone submodular cost.
For example, this holds when  $f$ is a rank function of a matroid.

\begin{restatable}{corollary}{corquasip}
\label{cor:quasi}
There is an algorithm that given an $n$-vertex graph $G$ and an integer monotone submodular function $f\colon 2^{V(G)}\rightarrow \mathbb{Z}_{\geq 0}$ represented by an oracle, finds a cycle $C$ in $G$ with $f(C) = \opt$ in time $n^{\Oh(\log \opt)}$.
\end{restatable}

Our second main result is that the running times of the algorithms of \Cref{thm:min-cycle} and \Cref{cor:quasi} are asymptotically tight.
Note that it is sufficient to prove \Cref{cor:quasi} to be tight, as any improvement to \Cref{thm:min-cycle} would also improve \Cref{cor:quasi}.

\begin{restatable}{theorem}{thmlowboundcycle}
\label{thm:lowe:bound:gen}
There is no algorithm computing a cycle of cost at most $\opt$ on a given $n$-vertex graph and an integer monotone submodular function $f : 2^{V(G)} \rightarrow \mathbb{Z}_{\ge 0}$ represented by an oracle, using at most $g(\opt) \cdot n^{o(\log \opt)}$ queries to the oracle, for any computable function $g$.
\end{restatable}

\begin{restatable}{corollary}{thmlowboundeps}
\label{thm:lowe:bound:eps}
There is no algorithm computing a cycle of cost at most $(1+\varepsilon) \cdot \opt$ on a given $n$-vertex graph and an integer monotone submodular function $f : 2^{V(G)} \rightarrow \mathbb{Z}_{\ge 0}$ represented by an oracle, using at most $t(1/\varepsilon) \cdot n^{o(\log 1/\varepsilon)}$ queries to the oracle, for any computable function $t$.
\end{restatable}


In particular, \Cref{thm:lowe:bound:gen} rules out fixed-parameter tractability (FPT) parameterized by $\opt$ and
\Cref{thm:lowe:bound:eps} rules out efficient polynomial-time approximation schemes (EPTAS).

The same construction as in \Cref{thm:lowe:bound:gen} also rules out the improvement of the quasipolynomial time in the setting where $\opt = \Oh(n)$.

\begin{restatable}{theorem}{thmlowboundquasip}
\label{the:noquasip}
There is no algorithm computing a cycle of cost at most $\opt = \Oh(n)$ on a given $n$-vertex graph and an integer monotone submodular function $f : 2^{V(G)} \rightarrow \mathbb{Z}_{\ge 0}$ represented by an oracle, using at most $n^{o(\log n)}$ queries to the oracle.
\end{restatable}

We note that on directed graphs the problem is much harder: The same construction as the one by Goel et al.~\cite{goel2009approximability} for undirected $(s,t)$-path shows that $\Oh(n^{2/3-\varepsilon})$-approximation for the minimum cost directed cycle requires an exponential number of queries to the oracle.

\Cref{thm:min-cycle} also yields a PTAS for computing the submodular connectivity of a planar multigraph. 
The \emph{connectivity} of a connected multigraph is the size of its minimum cut, that is, the minimum number of edges whose removal disconnects it.
In \textsc{Monotone Submodular Connectivity} (also known as \textsc{Monotone Submodular Min-Cut}), for a connected multigraph $G$ with monotone submodular cost function $f$ on $E(G)$, the task is to identify the minimum cost $f(C)$ of a cut $C \subseteq E(G)$.  
In a connected planar multigraph $G$, an edge set of every simple cycle of $G$ is an edge set of an inclusion minimal edge cut in the dual of $G$, and vice versa.
Thus by  \Cref{thm:min-cycle}, we have the following corollary.

\begin{corollary}
\label{cor:min-cut}
There is an algorithm that given a planar $m$-edge multigraph $G$, parameter $\varepsilon >0$, and a monotone submodular function $f : 2^{E(G)} \rightarrow \mathbb{R}_{\ge 0}$ represented by an oracle, finds a cut $C$ in $G$ with $f(C)\leq (1+\varepsilon)\cdot \opt$ in time $m^{\Oh(\log 1/\varepsilon)}$ (where $\opt$ is the minimum cost of a cut).
\end{corollary}

The same lower bounds of \Cref{thm:lowe:bound:gen} and \Cref{the:noquasip} apply also to this setting (with $n$ replaced by $m$), showing that \Cref{cor:min-cut} is optimal, because the graph we use for the lower bound is planar (in particular, it is a dual of a planar multigraph).
The best previously known upper bound on submodular connectivity on planar graphs is due to Jegelka  and Bilmes \cite{jegelka2009notes} who gave an 
 $\Oh( \sqrt{n})$-approximation for this problem. 
 
An interesting variant of submodular connectivity was considered by Ghaffari,  Karger,  and Panigrahi~\cite{ghaffari2017random}.
In the \textsc{Hedge Connectivity} problem, the edge set of a multigraph $G$ is partitioned into sets called \emph{hedges}.
The graph is \emph{$k$-hedge-connected} if it is necessary to remove at least $k$ edge sets (hedges) in order to disconnect $G$.
Ghaffari,  Karger,  and Panigrahi~\cite{ghaffari2017random} gave a PTAS of running time   $n^{\Oh(\log 1/\varepsilon)}$ and a quasipolynomial-time exact algorithm for hedge connectivity.
Very recently Jaffke et al.~\cite{JaffkeLMPS23} (see also~\cite{JaffkeLMPS22}) complemented this result by showing that the quasi-polynomial running time is  optimal up to the Exponential Time Hypothesis (ETH). Namely, they proved that 
 the existence of an algorithm with running time $(nk)^{o(\log n/(\log\log n)^2)}$ would contradict ETH.
The hedge function (i.e., the number of hedges covering an edge subset) is a monotone submodular function.
Thus, on planar graphs, \Cref{cor:min-cut} extends the PTAS of~\cite{ghaffari2017random} from hedges to monotone submodular functions.
Similarly, the quasipolynomial algorithm for integer-valued monotone submodular functions with $\opt=n^{\Oh(1)}$, extends the quasipolynomial exact algorithm of Ghaffari,  Karger,  and Panigrahi on planar graphs.
\medskip

While \Cref{thm:lowe:bound:gen} refutes the existence of a polynomial-time (or even FPT) algorithm computing a submodular minimum cycle or submodular minimum cut in planar graphs with polynomially bounded integer-valued functions, the complexity of the hedge variants of these problems remains open (here, by \emph{hedge minimum cycle} we mean the minimum number of hedges covering a cycle).
In graph theory, this problem is also known as the \textsc{Colored Cycle} problem \cite{BroersmaLWZ05}. 
In this reformulation of the problem, the edges (or vertices) of the given graph are colored and the task is to select a cycle containing the minimum number of different colors. 
Broersma et al. claimed the \textsc{Colored Cycle} problem to be NP-hard, without proof
\cite[Corollary~16]{BroersmaLWZ05}.
The quasipolynomial algorithm for this problem that follows by \Cref{cor:quasi} raises serious concerns about this claim. 
Note that the hedge minimum $(s,t)$-cut and hedge minimum $(s,t)$-path are indeed NP-hard~\cite{BroersmaLWZ05,zhang2011approximation}.

Motivated by the question on whether \textsc{Hedge Minimum Cycle} admits a polynomial-time algorithm or our quasipolynomial-time algorithm is optimal, we study the problem in a special case that corresponds to a natural problem about families of sets.
In particular, we consider the \textsc{Hedge Minimum Cycle} problem on the subdivisions of the graphs used for the lower bound construction of \Cref{thm:lowe:bound:gen} and \Cref{the:noquasip} -- see \Cref{fig:lbG}.
In these graphs, the \textsc{Hedge Minimum Cycle} problem is equivalent to the following set family problem:
For an integer $k$ and universe $U$, we say that a family $\mathcal{F}$ of sets over $U$ is \emph{$k$-wide} if for any two distinct sets $A,B \in \mathcal{F}$ it holds that $|A \cup B| > k$.
In the \textsc{Wide Family Hitting} problem, we are given a universe $U$, an integer $k$, and $m$ $k$-wide families $\mathcal{F}_1, \ldots, \mathcal{F}_m$.
The task is to decide if it is possible to select one set $S_i \in \mathcal{F}_i$ from each family $\mathcal{F}_i$ so that $|\bigcup_{i=1}^m S_i| \le k$.
We denote the input size by $N = \sum_{i=1}^m \sum_{A \in \mathcal{F}_i} |A|$.
The algorithm of \Cref{cor:quasi} gives an $N^{\Oh(\log k)}$ time algorithm for \textsc{Wide family hitting}, in particular it can be solved in quasipolynomial time, and therefore is unlikely to be NP-hard.

While it remains open whether \textsc{Wide Family Hitting} admits a polynomial-time algorithm, we show two results giving evidence that the special case of hedges is indeed easier than the general case of monotone submodular functions.
First, we show that \textsc{Wide Family Hitting} is fixed-parameter tractable when parameterized by $k$.
This is in contrast to the lower bound of \Cref{thm:lowe:bound:gen}.

\begin{restatable}{theorem}{thmwfhfpt}
\label{the:wfh-fpt}
There is a $2^{\Oh(k \log k)} N^{\Oh(1)}$ time algorithm for {\sc Wide Family Hitting}.
\end{restatable}

We then show that there is a polynomial-time algorithm if $|\mathcal{F}_i|$ is bounded for every $i$.
This corresponds to the case when the graph of the construction has bounded degree.

\begin{restatable}{theorem}{thmwfhbounddeg}
\label{the:wft-bound-deg}
Let $|\mathcal{F}_i| \le d$ for every $i$.
Then there is a $k^{\Oh(\log d)} N^{\Oh(1)}$ time randomized algorithm for {\sc Wide Family Hitting}.
\end{restatable}

The rest of the paper is organized as follows.
In \Cref{sec:prelim} we give formal definitions and preliminary results.
In \Cref{sec:min} we give the algorithm of \Cref{thm:min-cycle}.
In \Cref{section:lowerbound} we show the lower bounds \Cref{thm:lowe:bound:gen} and \Cref{the:noquasip}.
In \Cref{sec:wfhp} we prove \Cref{the:wfh-fpt} and \Cref{the:wft-bound-deg}.
We then conclude in \Cref{sec:concl}, in particular discussing open problems related to  \textsc{Hedge Minimum Cycle} and \textsc{Wide Family Hitting}.

\section{Preliminaries}\label{sec:prelim} 

In this section, we introduce basic notation used throughout the paper. 
%

We use standard graph-theoretic terminology and refer to the textbook of Diestel~\cite{Diestel12} for missing notions. 
 We consider only finite  graphs, and the considered graphs are assumed to be undirected if it is not explicitly said to be otherwise.  For  a graph $G$, we use  $V(G)$ and $E(G)$ to denote its vertex and edge set, respectively. 
 Throughout the paper we use $n=|V(G)|=|G|$ and $m=|E(G)|$.
 For a graph $G$ and a subset $X\subseteq V(G)$ of vertices, we write $G[X]$ to denote the subgraph of $G$ induced by $X$. 
For a vertex $v$, we denote by $N_G(v)$ the \emph{(open) neighborhood} of $v$, i.e., the set of vertices that are adjacent to $v$ in $G$. For $X\subseteq V(G)$, $N_G(X)=\big(\bigcup_{v\in X}N_G(v)\big)\setminus X$.
 The \emph{degree} of a vertex $v$ is $d_G(v)=|N_G(v)|$. We may omit subscripts if the considered graph is clear from a context.

A \emph{path} $P$ in $G$ is a subgraph of $G$ with $V(P)=\{v_0,\ldots,v_\ell\}$ and $E(P)=\{v_{i-1}v_i\mid1\leq i\leq \ell\}$.  
We write  $v_0v_1\cdots v_\ell$ to denote $P$; the vertices $v_0$ and $v_\ell$ are \emph{end-vertices} of $P$,  the vertices $v_1,\ldots,v_{\ell-1}$ are \emph{internal}, and $\ell$ is the \emph{length} of $P$. 
For a path $P$ with end-vertices $s$ and $t$, we say that $P$ is an $(s,t)$-path. 
A \emph{cycle} is a graph $C$ with $V(C)=\{v_1,\ldots,v_\ell\}$ for $\ell\geq 3$ and $E(C)=\{v_{i-1}v_i\mid1\leq i\leq \ell\}$, where we assume that $v_0=v_\ell$. We write $C=v_1\cdots v_\ell$ to denote a cycle in $G$.

\begin{definition}\label{def:submod}
 Given a finite set $U$, a function $f\colon 2^U\rightarrow \mathbb{R}$ is \emph{submodular} if for every $X,Y\subseteq U$,
 \[
 f(X)+f(Y)\geq f(X\cup Y)+f(X\cap Y). 
\]  
\end{definition}
We also will use an equivalent formulation of submodularity, that is, for any $X\subseteq Y$ and $v\not\in Y$, 
\[
f(X\cup \{v\})- f(X)\geq f(Y\cup \{v\})- f(Y).
\]

Throughout the paper we assume that the considered submodular functions $f\colon 2^U\rightarrow \mathbb{R}$
are given by value-giving oracles returning the value $f(X)$ for every $X\subseteq U$ in unit time.
We also assume the \emph{real RAM} computational model for operations with the values of considered functions, i.e., we assume that basic arithmetic operations over real numbers are performed in unit time. 
In this paper, we consider functions $f$ defined on subsets of the vertex or edge set of a graph. Slightly abusing notation, we may write $f(H)$ instead of $f(V(H))$ or $f(E(H))$ for a subgraph $H$ of $G$.

A submodular function is \emph{monotone} if for every $X\subseteq Y\subseteq U$, $f(X)\leq f(Y)$. We note that it is well-known that a rank function of a matroid is a monotone submodular function with nonnegative integer values.

 \section{PTAS for shortest cycles with monotone submodular costs}\label{sec:min}
 In this section, we demonstrate a PTAS for finding a shortest cycle with nonnegative monotone submodular costs. 
  If a connected component of a graph $G$ is a tree, it does not contain any cycle. In this case, the problem of finding a cycle in this component is meaningless. From now on, we assume that all connected components of graphs considered throughout the section  contain cycles. 
For a graph  $G$ and a function  $f\colon 2^{V(G)}\rightarrow \mathbb{R}_{\geq 0}$, we 
define   \[\opt(G,f)=\min\{f(C)\mid C\subseteq V(G)\text{ induces a cycle of }G\};\] 
 we write $\opt$ instead of $\opt(G,f)$ if $G$ and $f$ are clear from the contexts. 
 
First, we show that the problem admits a factor-2 approximation. Besides an approximate solution, our algorithm computes a family of induced tree-subgraphs rooted in the vertices of $G$ that will be crucial for PTAS. 
 
 \begin{lemma}\label{lem:const-fact} 
 There is an algorithm $\mathcal{A}$ that, given a graph $G$ and a monotone submodular function $f\colon 2^{V(G)}\rightarrow \mathbb{R}_{\geq 0}$,
 in time $\Oh(n(m+n\log n))$ finds a cycle $C$ with $f(C)\leq 2\opt$. Furthermore, the algorithm returns a family of induced tree-subgraphs
 $\cT_f=\{T_f(v)\}_{v\in V(G)}$ in $G$ such that for every $v\in V(G)$, (i) $v\in V(T_f(v))$ and (ii) for every $x\in V(T_f(v))$ and $y\in N_G(x)\setminus V(T_f(v))$, $f(Py)\geq \opt/2$, where $P$ is the unique $(v,x)$-path in $T_f(v)$. 
 \end{lemma}
 \begin{proof}
Our algorithm is based on the classical Dijkstra's algorithm for finding shortest paths~\cite{Dijkstra59}. Let $v\in V(G)$. The algorithm constructs a tree rooted in $v$ by assigning labels $p(x)$ for  vertices  $x\in V(G)$, where $p(x)$ is the parent of $x$ in the tree; initially $p(v)=v$ and $p(x)$ is empty for every $x\in V(G)$ distinct from $v$. For $x\in V(G)$ with nonempty $p(x)$, we use $P_x$ to denote  the unique $(v,x)$-path defined by these labels. We also assign labels $d(x)$ for $x\in V(G)$, where $d(x)=f(P_x)$ if $p(x)$ is nonempty. Then the following subroutine computes a cycle $C_v$ associated with $v$ and $T_f(v)$ defined by the set of vertices given together with their labels $p(x)$. 

\medskip
\begin{algorithm}[H]
\KwIn{A graph $G$ with $v\in V(G)$ and a function $f$.}
\KwResult{A cycle $C_v$ and a tree $T_f(v)$.}
\Begin
{ 
set $S:=V(G)$, $p(v):=v$, $d(v):=f(v)$\;
\ForEach{$x\in V(G)\setminus\{v\}$}
{set $p(x):=\emptyset$ and $d(x):=+\infty$
}  
\While{$S\neq\emptyset$}
{
find $x\in S$ s.t. $d(x)=\min\{d(y)\colon y\in S\}$ and set $S:=S\setminus\{x\}$\;
\eIf{there is $y\in N_G(x)\setminus\{p(x)\}$ with $d(y)\leq d(x)$ \label{ln:cycle}}{
find a cycle $C_v$ in $G[V(P_x)\cup V(P_y)]$ and output $C_v$\;
output $T_f(v)$ with the set of vertices $\{z\in V(G)\colon d(z)< d(x)\}$\label{ln:tree}\;
\textbf{quit}
}
{\ForEach{$y\in N_G(x)\setminus\{p(x)\}$ with $d(y)>f(P_xy)$}
{set $d(y):=f(P_xy)$ and $p(y):=x$}
}
}
}
\caption{$\textsc{Cycle}(G,v,f)$} \label{alg:min-rank}
\end{algorithm}

\medskip
To analyze the algorithm, denote by $g(x)=\min\{f(P)\colon P\text{ is a }(v,x)\text{-path in }G\}$ for every $x\in V(G)$. Clearly, $g(x)\leq d(x)$ for $x\in V(G)$. 
For a real number $h\geq f(v)$, let $G_h$ be the subgraph of $G$ induced by the set of vertices $\{x\in V(G)\colon g(x)\leq h\}$.
Let $h^*\geq f(v)$ be the minimum number such that $G_{h^*}$ contains a cycle. Notice that such a number exists, because the connected component of $G$ containing $v$ is not a tree. 
Note also that for every $h<h^*$, $G_h$ is a tree.
The crucial observation is that the algorithm assigns the labels $d(x)=g(x)$ for $x\in V(G_h)$ if $h<h^*$ and the labels $p(x)$ define the induced tree $G_h$. Furthermore, the algorithm stops 
in line~(\ref{ln:cycle}), where $d(x)=g(x)=h^*$ and $d(y)=g(y)\leq h^*$. Because $xy\in E(G)$ and $y\neq p(x)$, the graph $G[V(P_x)\cup V(P_y)]$ contains a cycle $C_v$. Because $f(P_x)=g(x)$ and $f(P_y)=g(y)$, we have that $f(C_v)\leq 2h^*$. Since $d(x)=h^*$, we have that $T_f(v)$ constructed in line~(\ref{ln:tree}) is an induced tree in $G$. 

Clearly, $v\in V(T_f(v))$, and condition (i) for $T_f(v)$ is fulfilled. 
By definition, $f(C_v)\geq \opt$. Hence, $h^*\geq \opt/2$. If there are $x\in V(T_f(v))$ and $y\in N_G(v)\setminus V(T_f(v))$ such that $f(P_xy)<\opt/2$, then $g(y)<h^*$ and $y$ should be in $T_f(v)$. This implies that (ii) holds.

We run $\textsc{Cycle}(G,v,f)$ for all $v\in V(G)$ and construct $\cT_f=\{T_f(v)\}_{v\in V(G)}$. To find $C$, we consider the cycles $C_v$ for $v\in V(G)$ and select a cycle $C$ of minimum cost. To show that $f(C)\leq 2\opt$, consider $v\in V(C)$. 
Then $C$ contains a $(v,y)$-path $P=P_xy$, where $x\in V(T_f(v))$ and $y$ is adjacent to $x$. Then $f(P)\geq h^*$ for $h^*$ defined for this vertex $v$. Because $C=C_v$ and $f(C_v)\leq 2h^*$, $f(C)\leq 2\opt$.

To evaluate the running time, note that Dijkstra's algorithm can be implemented to run in $\Oh(m+n\log n)$ time by the results of Fredman and Tarjan~\cite{FredmanT84}. Using exactly the same approach, we conclude that for each $v\in V(G)$,  $\textsc{Cycle}(G,v,f)$ can be implemented to run in $\Oh(m+n\log n)$ time. Since the algorithm is called for every $v\in V(G)$, the total running time is $\Oh(n(m+n\log n))$. This concludes the proof. 
\end{proof}
 
Let  $\cT=\{T(v)\}_{v\in V(G)}$ be a family of induced tree-subgraphs in a graph $G$ such that  $v\in V(T(v))$ for every $v\in V(G)$. 
For $v\in V(G)$, we define the family of paths 
\begin{equation}\label{eq:paths}
\cP(v)=\{Py\colon P\text{ is a }(v,x)\text{-path for }x\in V(T(v))\text{ and }y\in N_G(x)\setminus V(T(v))\}, 
\end{equation}
and set $\cP(\cT)=\bigcup_{v\in V(G)}\cP(v)$. We use the following easy property of these paths. 
 
 \begin{lemma}\label{lem:branch-path}
 Let $\cP(\cT)$ be the family of paths constructed for $\cT=\{T(v)\}_{v\in V(G)}$. Then for every cycle $C$, there is a path $P\in \cP$ such that $P$ is a segment of $C$. Furthermore, $|\cP(\cT)|\leq nm$ and the sets of vertices of the paths of $\cP(\cT)$ can be listed in $\Oh(n^2m)$ time.
  \end{lemma}
 
 \begin{proof}
 Let $\cP(\cT)=\bigcup_{v\in V(G)}\cP(v)$, where $\cP(v)$ is defined as in (\ref{eq:paths}).
  Consider a vertex $v\in V(C)$. Because $T(v)$ is an induced tree in $G$, $C$ contains a path $Py$, where $P$ is a $(v,x)$-path $P$ in $T(v)$ and $y\in N_G(x)\setminus V(T(v))$. By definition, $Py\in\cP(v)$. This proves that $C$ contains as a segment a path from $\cP(\cT)$. Since 
 every vertex $y\in V(G)\setminus V(T(v))$ has at most $\deg_G(y)$ neighbors in $T(v)$, the number of paths in  $\cP(v)$ does not exceed $m$. Hence,  $|\cP(\cT)|\leq nm$.
 To list the set of vertices of the paths of $\cP(v)$, we consider every vertex $y\in V(G)\setminus V(T(v))$ and for each neighbor $x$ in $T(v)$, we trace the unique $(x,v)$-path with at most $n$ vertices. Therefore,
 the sets of vertices of the paths of $\cP(\cT)$ can listed in $\Oh(n^2m)$ time. 
 \end{proof}
 
 We are ready to prove \Cref{thm:min-cycle}, which we restate here.
  
  \thmmaincycle*
 
\begin{proof}
The rough idea is that we construct a recursive branching algorithm using \Cref{lem:const-fact} and \Cref{lem:branch-path}. In particular, the algorithm from \Cref{lem:const-fact} constructs a family of induced trees $\cT_f$. Then by \Cref{lem:branch-path}, a solution cycle $C$ should contain some path $P\in\cP(\cT_f)$ as a segment. We branch on these paths. However, instead of looking for a cycle containing $P$, we simply redefine the function by setting 
 $g(X)=f(X\cup V(P))-f(P)$ for each  $X\subseteq V(G)$ using the property that for any cycle $C$, $f(C)\leq f(V(C)\cup V(P))=g(C)+f(P)$ and $f(C)=g(C)+f(P)$ if $V(P)\subseteq V(C)$.
 Then we solve  the problem recursively for the new function.  Because $f(P)\geq\opt/2$ by \Cref{lem:const-fact}, we require a logarithmic in $1/\varepsilon$ depth of the search tree before we can apply a 2-approximation from \Cref{lem:const-fact} to obtain a factor-$(1+\varepsilon)$ approximation.   
 
To describe the algorithm formally, we construct the subroutine $\textsc{Find-Cycle}(G,g,k)$, which 
takes as its input $G$, a monotone submodular  function $g\colon 2^{V(G)}\rightarrow \mathbb{R}_{\geq 0}$,  and an integer $k\geq 0$.
The subroutine returns  a cycle $C$ of $G$ with $g(C)\leq\big(1+\frac{1}{2^k}\big)\opt(G,g)$.
Initially, $g:=f$.  The parameter $k$ defines the depth of recursion and is initially set to $k:=\lceil\log 1/\varepsilon\rceil$. 
To solve the problem for our original instance, we call $\textsc{Find-Cycle}(G,f,\lceil\log 1/\varepsilon\rceil)$.
Recall that we use  $\mathcal{A}$ to denote the algorithm from \Cref{lem:const-fact}.

\medskip
\begin{algorithm}[H]
\KwIn{A graph $G$, function $g$, and  $k\geq 0$.}
\KwResult{A cycle $C$ of $G$.}
\Begin
{ 
call $\mathcal{A}(G,g)$ to obtain a cycle $C$ and a family of subtrees $\cT_g$\label{ln:start}\;
\If {$g(C)>0$ and $k>0$}
{ 
construct $\cP=\cP(\cT_g)$\label{ln:p}\;
\ForEach{$P\in\cP$ \label{ln:bloop}}
{
set $g'(X):=g(V(P)\cup X)-g(P)$ for $X\subseteq V(G)$ \label{ln:gprime}\;
call $\textsc{Find-Cycle}(G,g',k-1)$ to find a cycle $C'$ \label{ln:rec}\;
\If{$g(C')<g(C)$ \label{ln:if}}{set $C:=C'$ \label{ln:if-two}}
}\label{ln:eloop}
}
return $C$
}
\caption{$\textsc{Find-Cycle}(G,g,k)$} \label{alg:min-rank2}
\end{algorithm}

\medskip
To show correctness, note that if $g\colon 2^{V(G)}\rightarrow \mathbb{R}_{\geq 2}$ is a monotone submodular function, then 
 each function $g'$ introduced in line (\ref{ln:gprime}) is also a  monotone submodular function with nonnegative values, that is, the input $\textsc{Find-Cycle}(G,g',k-1)$ in line (\ref{ln:rec}) is feasible. Further,  $\textsc{Find-Cycle}(G,g,k)$ is finite, because the depth of the recursion is upper bounded by $k$. Also the subroutine algorithm always returns some cycle of $G$ because $G$ is distinct from a forest by our assumption. Hence, to prove the correctness of $\textsc{Find-Cycle}(G,g,k)$, we have to show that it returns a cycle $C$ with $g(C)\leq\big(1+\frac{1}{2^k}\big)\opt(G,g)$.
We show this  by induction on $k$.   

If $k=0$, then the algorithm returns the cycle $C$ produced by $\mathcal{A}(G,g)$ and, therefore, $g(C)\leq 2\opt(G,g)=\big(1+\frac{1}{2^k}\big)\opt(G,g)$. 
Let $k>0$ and assume that  $\textsc{Find-Cycle}(G,g',k-1)$ called in line (\ref{ln:rec}) outputs $C'$ with $g'(C')\leq \big(1+\frac{1}{2^{k-1}}\big)\opt(G,g')$.

If $g(C)=0$ for the cycle $C$ constructed by $\mathcal{A}(G,g)$ in line (\ref{ln:start}), then the claim is trivial. Assume that $g(C)>0$.
Let $C^*$ be a cycle of $G$ with $g(C^*)=\opt(G,g)$. By \Cref{lem:const-fact} and \Cref{lem:branch-path}, there is $P\in \cP(\cT_g)$ such that $P$ is a segment of $C^*$ and $g(P)\geq\opt(G,g)/2$. We consider $P$ in the loop in lines (\ref{ln:bloop})-(\ref{ln:eloop}). Then, for the function $g'$ considered in line (\ref{ln:gprime}), 
\begin{equation*}
g'(C^*)=g(C^*)-g(P)\leq g(C^*)-\opt(G,g)/2=\opt(G,g)/2. 
\end{equation*}
Therefore, 
\begin{equation}\label{eq:rec}
\opt(G,g')\leq \opt(G,g)-g(P)\text{ and }\opt(G,g')\leq \opt(G,g)/2. 
\end{equation}
Let $C'$ be the cycle produced by $\textsc{Find-Cycle}(G,g',k-1)$ in line (\ref{ln:rec}).
By the inductive assumption 
\begin{equation*}
g'(C')\leq \big(1+\frac{1}{2^{k-1}}\big)\opt(G,g').
\end{equation*} 
Then by the definition of $g$ and (\ref{eq:rec}),
\begin{align*}
g(C')\leq &~ g(V(C')\cup V(P))=g'(C')+g(P)\leq \big(1+\frac{1}{2^{k-1}}\big)\opt(G,g')+g(P)\\
=&~(\opt(G,g')+g(P))+\frac{1}{2^{k-1}}\opt(G,g')\\
\leq&~\opt(G,g)+\frac{1}{2^{k}}\opt(G,g)=\big(1+\frac{1}{2^{k}}\big)\opt(G,g).
\end{align*}
By the choice of $C$ in lines (\ref{ln:if})--(\ref{ln:if-two}), the algorithm outputs a cycle $C$ with 
$g(C)\leq g(C')\leq \big(1+\frac{1}{2^{k}}\big)\opt(G,g)$. This concludes the correctness proof.

We call $\textsc{Find-Cycle}(G,f,k)$, where $k=\lceil\log 1/\varepsilon\rceil$, to solve the problem for the original instance. 
Because the algorithm outputs a cycle $C$ with 
$f(C)\leq\big(1+\frac{1}{2^k}\big)\opt(G,f)$ and  $k=\lceil\log 1/\varepsilon\rceil$, 
$f(C)\leq (1+\varepsilon)\opt(G,f)$, that is, we obtain the desired approximation. 

To evaluate the running time, note first that  we switch to the function $g$ in line (\ref{ln:gprime}). 
We can make the following easy observation about such functions. Suppose that $f_1,f_2,f_3\colon 2^{V(G)}\rightarrow \mathbb{R}_{\geq 0}$ are functions such that for every $X\subseteq V(G)$, $f_2(X)=f_1(X\cup A)-f_1(A)$ and $f_3(X)=f_2(X\cup B)-f_2(B)$ for some $A,B\subseteq V(G)$. 
Then 
\begin{align*}
f_3(X)=&f_2(X\cup B)-f_2(B)=(f_1(X\cup B\cup A)-f_1(A))-(f_1(A\cup B)-f_1(A))\\
=&f_1(X\cup (B\cup A))-f_1(A\cup B).  
\end{align*}
Using this observation iteratively,  using only the oracle for the input function $f$, the values of all other functions occurring in the algorithm could  be computed in $\Oh(n)$ time for each $X\subseteq V(G)$. 

Computing  $C$ and $\cT_g$ in line~(\ref{ln:start}) can be done in  $\Oh(n^2(m+n\log n))$ time by \Cref{lem:const-fact} taking into account that each value $g(X)$ can be computed in $\Oh(n)$ time. The construction of $\cP(\cT_g)$ can be done in $\Oh(n^2m)$ time by \Cref{lem:branch-path}. The number of paths $P$ considered in the loop in lines~(\ref{ln:bloop})--(\ref{ln:eloop}) is at most $nm$ by \Cref{lem:branch-path}. Therefore, the number of recursive calls of $\textsc{Find-Cycle}(G,g',k-1)$ in line (\ref{ln:rec}) is at most $nm\leq n^3$.
The depth of the search tree is at most $\lceil\log 1/\varepsilon \rceil$.  Therefore, the total running time is $n^{\Oh(\log 1/\varepsilon)}$. 
This concludes the proof.
\end{proof}
 
 When the function $f$ is integer-valued,  \Cref{thm:min-cycle} (by setting  $\varepsilon = \frac{1}{w +1}$ with $\opt\leq w\leq 2\opt$, where $w=f(C)$ for the cycle $C$ returned by the approximation algorithm for $\varepsilon=1/2$)     implies 
that a cycle of cost $\opt$ can be found in time $n^{\Oh(\log \opt)}$.
In particular, when 
$\opt=n^{\Oh(1)}$, we obtain a quasi-polynomial algorithm computing the cycle of minimum submodular cost. For example, this holds if $f$ is a rank function of a matroid.

\corquasip*

 
Finally in this section, we observe that our results can be easily translated for the edge version of the problem, even on multigraphs.
For monotone submodular function $f\colon 2^{E(G)}\rightarrow \mathbb{R}_{\geq 0}$, we define 
$\opt=\min\{f(C)\colon C \subseteq E(G) \text{ is a cycle of  }G\}$ in the same way as for the vertex costs.

  \begin{corollary}\label{cor:min-cycle}
Let $G$ be an $m$-edge multigraph, $\varepsilon>0$, and 
$f\colon 2^{E(G)}\rightarrow \mathbb{R}_{\geq 0}$ a monotone submodular function represented by an oracle.
Then a cycle $C$ in $G$ with $f(C)\leq (1+\varepsilon)\cdot \opt$ can be found in time $m^{\Oh(\log 1/\varepsilon)}$.
\end{corollary}
\begin{proof}
We construct a graph $G'$ by subdividing each edge of $G$ once, that is, for each edge $xy\in E(G)$, we introduce a new vertex $v_{xy}$, make $v_{xy}$ adjacent to $x$ and $y$, and delete $xy$.
For a subdivision vertex $v_{xy}$, define $e(v_{xy})=xy$.  
Let $W$ be the set of subdivision vertices.
We define $g\colon 2^{V(G')}\rightarrow\mathbb{R}_{\geq 0}$ by setting 
$g(X)=f(\{e(v)\colon v\in W\cap X\})$ for each $X\subseteq V(G')$.
The definition implies that  $g\colon 2^{V(G')}\rightarrow \mathbb{R}_{\geq 0}$ is a monotone submodular function and that an oracle for $f$ can be translated into an oracle for $g$.
There is one-to-one correspondence between cycles of $G$ and $G'$, because each cycle $C'$ is obtained from a cycle $C$ of $G$ by subdividing edges and $f(C)=g(C')$.
 Therefore, we can apply \Cref{thm:min-cycle} for $G'$ and $g$.
\end{proof}

\section{Lower bound}\label{section:lowerbound}
In this section we prove the lower bounds of \Cref{thm:lowe:bound:gen} and \Cref{the:noquasip}.
Both of these lower bounds will follow from the same construction, although with different parameters.

We give the lower bounds for the setting where the function $f$ is defined on the edges of a multigraph, which then by \Cref{cor:min-cycle} translates into a lower bound when the function is defined on vertices of a graph.
In our construction the function $f$ is integer-valued.

Our lower bound is based on the following construction. 
For positive integers $k$ and $p$ we  define a multigraph \Ggp with $k+1$ vertices and $pk+1$ edges (see \Cref{fig:lbG}) and a monotone submodular function $f\colon 2^{E(G(k,p))}\rightarrow \mathbb{N}$ so that $\opt(G(k,p),f) = 2^{k+1}-1$.
Then, for each cycle $C$ of length $k+1$ of \Ggp we define a monotone submodular function $f_C \colon 2^{E(G(k,p))} \rightarrow \mathbb{N}$ so that $\opt(G(k,p),f_C) = 2^{k+1}-2$ and $f_C$ differs from $f$ only on the cycle $C$.
Deciding whether an oracle represents the function $f$ or one of the functions $f_C$ will then require querying each cycle $C$ of length $k+1$ and there are $p^k$ such cycles in $G(k,p)$.

\medskip\noindent\textbf{Construction of \Ggp.}
The multigraph \Ggp has vertex set $ \{v_1, v_2, \dots, v_{k+1}\}$.
For every pair of consecutive vertices $v_i, v_{i+1}$, $1\leq i\leq k$, there are $p$ parallel edges $F_i=\{e^1_i, \dots, e^p_i\}$ with endpoints $v_i$ and $v_{i+1}$.
One more edge $e_{k+1}$ connects $v_1$ and $v_{k+1}$, see \Cref{fig:lbG}.
In total,  \Ggp has $k+1$ vertices and $m=pk+1$ edges.
The multigraph \Ggp contains $p^k$ cycles of length~$k+1$.
Each such cycle passes through all the vertices of the multigraph in the order $v_1, v_2, \dots, v_{k+1}, v_1$.

\begin{figure} 
\begin{center}
\begin{tikzpicture}
\node[circle,draw=black,fill=white,inner sep= 2pt,minimum width = 3pt,label={below:$v_1$}] (v1) at (0,0) {};

\node[circle,draw=black,fill=white,inner sep= 2pt,minimum width = 3pt,label={below:$v_2$}] (v2) at (3,0) {};

\node[circle,draw=black,fill=white,inner sep= 2pt, minimum width = 1pt] (v3) at (6,0) {};

\node[circle,draw=black,fill=white,inner sep= 2pt,minimum width = 1pt] (vk) at (8,0) {};

\node[circle,draw=black,fill=white,inner sep= 2pt,minimum width = 1pt,label={below right:$v_{k+1}$}] (vk1) at (11,0) {};

\foreach \i in {20,35,50}{
\draw (v1) to [bend right = \i] (v2);
\draw (v1) to [bend left = \i] (v2);
\draw (v2) to [bend right = \i] (v3);
\draw (v2) to [bend left = \i] (v3);
\draw (vk) to [bend right = \i] (vk1);
\draw (vk) to [bend left = \i] (vk1);
}
\draw (v1) to [bend right = 50] node[pos=0.5,label={above:$e_{k+1}$}] {} (vk1);

\draw [thick,decorate,
    decoration = {brace}] (-0.5,-0.65) to node[pos=0.5,label={left:$p$}] {}  (-0.5,0.65);
\node[label={$\vdots$}] () at (1.5,-0.4) {};
\node[label={$\vdots$}] () at (4.5,-0.4) {};
\node[label={$\cdots$}] () at (7,-0.4) {};
\node[label={$\vdots$}] () at (9.5,-0.4) {};

\end{tikzpicture}
\caption{Construction of the multigraph \Ggp.}\label{fig:lbG}
\end{center}
\end{figure}
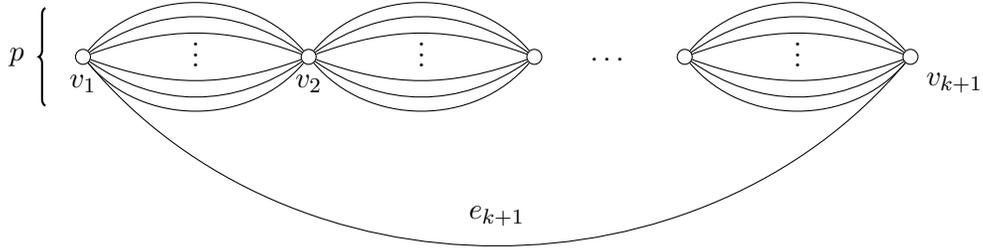

\medskip\noindent\textbf{Constructions of $f$ and $f_C$.}
We define the following function $f$ on the subsets $X$ of $E(G(k,p))$.
First, if $X \subseteq E(G(k,p))$ contains a cycle, i.e., $|X| \ge k+1$ or there is $i$ so that $|X \cap F_i| \ge 2$, we define
\[
f(X)=2^{k+1}-1.
\] 
Otherwise, i.e., $X \subseteq E(G(k,p))$ does not contain a cycle, we define
\[
f(X)=2^{k+1}- 2^{k+1-|X|},
\]
and by definition we have that $\opt(G(k,p),f) = 2^{k+1}-1$.

For a cycle $C \subseteq E(G(k,p))$ of length $|C| = k+1$, the function $f_C$ is defined as $f_C(X) = f(X)$ for $X \neq C$, and $f_C(C) = 2^{k+1}-2$.
Clearly $\opt(G(k,p),f_C) = 2^{k+1}-2$, and this optimum is given uniquely by the cycle $C$.

It is clear from the definitions that the functions $f$ and $f_C$ are monotone.
Next we establish the submodularities of $f$ and $f_C$.
Note that it is sufficient to prove that $f_C$ is submodular, as then the submodularity of $f$ follows by writing $f$ as a restriction of $f_C$ on $G(k,p+1)$.

\begin{lemma}\label{lemma:fsubmodular}
The function $f_C$ is submodular.
\end{lemma}
\begin{proof}
To prove the submodularity of $f_C$, we show that for every two sets $X\subset Y \subset E(G(k,p))$ and $e \in E(G(k,p)) \setminus Y$, 
\begin{equation}\label{eq:submodularityG}
f_C(X\cup \{e\})- f_C(X)\geq f_C(Y\cup \{e\})- f_C(Y).
\end{equation}

Depending on $X$, $Y$, and $e$, we consider different cases.

\medskip\noindent\textsl{Case~1: $f_C(X\cup \{e\}) = 2^{k+1}-1$.} 
Then also $f_C(Y\cup \{e\}) = 2^{k+1}-1$, and as $f_C(X) \le f_C(Y)$ by monotonicity, \eqref{eq:submodularityG} follows.

\medskip\noindent\textsl{Case~2: $f_C(Y) = 2^{k+1}-1$.} In this case
$
f_C(Y)=f_C(Y\cup \{e\}),
$
and  \eqref{eq:submodularityG} follows by the monotonicity of $f_C$.

\medskip\noindent\textsl{Case~3: $f_C(Y\cup \{e\}) = 2^{k+1}-1$.} If either $f_C(X\cup \{e\}) = 2^{k+1}-1$ or $f_C(Y) = 2^{k+1}-1$, then we are done by the previous cases.
Otherwise, $|X \cup \{e\}| \le |Y| \le k+1$, and we consider two subcases.
 
\medskip\noindent\textsl{Subcase~3a: $X\cup \{e\} = C$.}
Then $|X|=k$ and hence $|Y| \geq k+1$.
In this case because $Y \neq X \cup \{e\}$,
$  
f_C(Y)=f_C(Y\cup \{e\})=2^{k+1}-1,
$
while 
$f_C(X)\leq f_C(X\cup \{e\})$
  by the monotonicity of $f_C$.
  
\medskip\noindent\textsl{Subcase~3b: $X\cup \{e\} \neq C$.} In this case, $X \cup \{e\}$ does not contain a cycle, and therefore we have that
 $
 f_C(X\cup \{e\})-  f_C(X)
 =2^{k+1}-2^{k+1-|X|-1}
-(2^{k+1}-2^{k+1-|X|})= 2^{k -|X|}
$
and $|X| \le k-1$.
Then, if $Y = C$, we have that $f_C(Y \cup \{e\}) - f_C(Y) = 1 \le 2^{k-|X|}$.
If $Y \neq C$, then $Y$ does not contain a cycle and we have that
$
 f_C(Y\cup \{e\})-  f_C(Y)
 =2^{k+1}-1 
-(2^{k+1}-2^{k+1-|Y|})< 2^{k+1 -|Y|}\leq 2^{k -|X|}.
$ (For the last inequality we use $|X| < |Y|$.)

\medskip\noindent\textsl{Case~4: None of the previous cases holds.} In this case $X \cup \{e\}$ does not contain a cycle, so we have that
$ f_C(X\cup \{e\})-  f_C(X)= 2^{k -|X|}$. If $Y\cup \{e\} = C$, then 
$f_C(Y\cup \{e\})-  f_C(Y)= 2^{k+1} -2 -(2^{k+1} -2^{k+1-k})=0$.
If $Y\cup \{e\} \neq C$, 
then $ f_C(Y\cup \{e\})-  f_C(Y) = 2^{k -|Y|} \leq 2^{k -|X|}$. 
\end{proof}

Now each of the functions $f_C$ and the function $f$ could be represented by the oracle, and the optimum depends on whether the function represented by the oracle is $f$ or one of the functions $f_C$.
Therefore, it remains to argue that we cannot distinguish between $f$ or one of $f_C$ in less than $p^k$ queries.

\begin{lemma}\label{lem:pumpkinsLB}
Let $g \colon 2^{E(G(k,p))} \rightarrow \mathbb{N}$ be a function represented by an oracle, with a promise that either $g = f$ or $g = f_C$ for some cycle $C$ of $G(k,p)$ of length $|C| = k+1$.
It requires at least $p^k$ queries to the oracle to determine if $g = f$.
\end{lemma}
\begin{proof}
Suppose the oracle answers the queries always according to the function $f$, and an algorithm terminates after asking less than $p^k$ queries.
Because $G(k,p)$ has $p^k$ cycles of length $k+1$, there exists some cycle $C$ so that the algorithm has not queried $C$, and therefore as $f$ and $f_C$ are equivalent on all inputs except $C$, all the answers are consistent with both $f$ and $f_C$.
Therefore the algorithm cannot decide correctly whether $g = f$ or $g = f_C$.
\end{proof}

Next we summarize the lower bound that follows from the constructions of the multigraph \Ggp, the functions $f$, and $f_C$, and \Cref{lem:pumpkinsLB}.

\begin{lemma}
\label{lem:LBMAIN}
For any positive integers $p,k$, there exists a graph $G$ with $kp+k+2$ vertices and an integer submodular function $f : 2^{V(G)} \rightarrow \mathbb{N}$ represented by an oracle so that deciding whether $\opt(G,f) = 2^{k+1}-2$ or $\opt(G,f) = 2^{k+1}-1$ requires at least $p^k$ queries to the oracle.
\end{lemma}
\begin{proof}
We take the multigraph \Ggp with $m=kp+1$ edges and let $f \colon 2^{E(G(k,p))} \rightarrow \mathbb{N}$ be a function represented by an oracle.
By \Cref{lem:pumpkinsLB}, deciding whether $\opt(G(k,p),f) = 2^{k+1}-2$ or $\opt(G(k,p),f) = 2^{k+1}-1$ requires $p^k$ queries  to the oracle in the worst case.
This construction is for a multigraph where the function is on the edges, but by the argument of \Cref{cor:min-cycle} the lower bound also holds for graphs with $kp+k+2$ vertices where the function is on the vertices.
\end{proof}

By making use of \Cref{lem:LBMAIN}, we establish \Cref{thm:lowe:bound:gen} with  the lower bound matching the algorithmic bound of \Cref{thm:min-cycle}. 
We restate the theorem here. 
 
\thmlowboundcycle*
\begin{proof}
We assume without loss of generality that $g$ is non-decreasing and $g(x) \ge x$ for every $x\in \mathbb{R}_{\geq 0}$.
Assume that there is an algorithm that makes at most $t(\opt, n) = g(\opt) \cdot n^{o(\log \opt)}$ queries.
Now, there exists some large enough $N$ and $k'$ so that $t(\opt, n) < g(\opt) \cdot n^{(\log_2 \opt)/16}$ for all $\opt \ge k'$ and $n \ge N$.
We apply \Cref{lem:LBMAIN} with $p = g(4k') \cdot N$ and $k = \lceil \log_2 k' \rceil$. Let $n=kp+k+2=g(4k')N+ \lceil \log_2 k' \rceil+2$. Because $g$ is non-decreasing and $g(x) \ge x$ for every $x\in \mathbb{R}_{\geq 0}$, we have that $\lceil \log_2 k' \rceil\leq k'\leq g(4k')$ and it holds that  $N \le n \le 3g(4k')N\leq (g(4 k') \cdot N)^2$ if $N\geq 3$.
We get a graph with $n$ vertices, where $N \le n \le (g(4 k') \cdot N)^2$, and optimum $\opt$ with $k' \le \opt \le 4k'$ in which the problem requires at least
\[
(g(4 k') \cdot N)^{\lceil \log_2 k' \rceil} \ge g(4 k') \cdot (g(4 k') \cdot N)^{\lceil \log_2 k' \rceil - 1} \ge g(\opt) \cdot n^{(\log_2 \opt)/16}
\]
queries to solve.
This contradicts the existence of such an algorithm.
\end{proof}

We then establish \Cref{the:noquasip}.

\thmlowboundquasip*
\begin{proof}
Assume there is an algorithm that makes at most $t(n) = n^{o(\log n)}$ queries.
Now, there exists a large enough $N$ so that $t(n) < n^{(\log_2 n) / 4}$ for all $n \ge N$.
However, applying \Cref{lem:LBMAIN} with $p = N$ and $k = \lceil \log_2 N \rceil$ gives a graph with $n$ vertices, where $N \le n \le N^2$, and optimum $\opt \le 4n$, in which the problem requires at least $N^{\lceil \log_2 N \rceil} \ge n^{(\log_2 n) / 4}$
queries to solve.
This contradicts the existence of such an algorithm.
\end{proof}

\section{The wide family hitting problem}
\label{sec:wfhp}
Motivated by the question whether \textsc{Hedge Minimum Cycle} admits a polynomial-time algorithm, and the fact that our algorithm for monotone submodular functions is optimal already on a very restricted class of graphs  considered in Section~\ref{section:lowerbound}, we study the complexity of \textsc{Hedge Minimum Cycle} on the subdivisions of $G(k,p)$ (see \Cref{fig:lbG}).
In this class, \textsc{Hedge Minimum Cycle} is equivalent to a problem which we call \textsc{Wide Family Hitting}.

For an integer $k$ and a universe $U$, we say that a family $\mathcal{F}$ of sets is \emph{$k$-wide} if for any two distinct sets $A,B \in \mathcal{F}$ it holds that $|A \cup B| > k$.
In the \textsc{Wide Family Hitting} problem, the input consists of an integer $k$, a universe $U$, and $m$ $k$-wide families $\mathcal{F}_1, \ldots, \mathcal{F}_m$ over the universe $U$.
The task is to decide if it is possible to select one set from each family, i.e., sets $S_1 \in \mathcal{F}_1, S_2 \in \mathcal{F}_2, \ldots, S_m \in \mathcal{F}_m$ so that $|\bigcup_{i=1}^m S_i| \le k$.
We denote the size of the input by $N = \sum_{i=1}^m \sum_{A \in \mathcal{F}_i} |A|$.

To see the relations between \textsc{Hedge Minimum Cycle} and \textsc{Wide Family Hitting}, we first show reduction from  \textsc{Wide Family Hitting} to  \textsc{Hedge Minimum Cycle}. 
Consider  $m$ $k$-wide families $\mathcal{F}_1, \ldots, \mathcal{F}_m$ over the universe $U$. We construct the vertex-colored graph $G$, where the vertices are colored by the elements of $U$, as follows:
\begin{itemize}
\item construct $m+1$ vertices $v_0,\ldots,v_m$ and color them by a special  color $c\notin U$;
\item for each $i\in\{1,\ldots,m\}$, construct $|\mathcal{F}_i|$ $(v_{i-1},v_i)$-paths such that for every $S\in \mathcal{F}_i$, we have a path with $|S|$ internal vertices colored by the elements of $S\subseteq U$;
\item make $v_0$ and $v_m$ adjacent. 
\end{itemize}
Let also $k'=k+1$. It can be seen that $G$ has a cycle $C$, whose vertices are colored by at most $k'$ colors, if and only if there are $S_i\in\mathcal{F}_i$ for $i\in\{1,\ldots,m\}$ such that $|\bigcup_{i=1}^m S_i| \le k$. To prove this, notice that because $\mathcal{F}_1, \ldots, \mathcal{F}_m$ are $k$-wide, any cycle $C$ in $G$ containing vertices of at most $k'=k+1$ colors should contain $(v_{i-1},v_i)$-paths for each $i\in\{1,\ldots,m\}$. For every $i\in\{1,\ldots,m\}$, let $S_i\in\mathcal{F}_i$ be the set of colors of the internal vertices of the $(v_{i-1},v_i)$-path in $C$. If $C$ contains vertices of at most $k'=k+1$ colors, then $|\bigcup_{i=1}^m S_i| \le k$. For the opposite direction, let $S_i\in\mathcal{F}_i$ for $i\in\{1,\ldots,m\}$ be such that $|\bigcup_{i=1}^m S_i| \le k$. Then we construct the cycle $C$ in $G$ by concatenating the $(v_{i-1},v_i)$-paths whose internal  vertices are colored by the elements of $S_i$ and completing the cycle by the addition of the edge $v_{0}v_m$. 
Clearly, the vertices of $C$ are colored by at most $k'=k+1$ colors. 

To reduce  \textsc{Hedge Minimum Cycle} to \textsc{Wide Family Hitting} on subdivisions of the graphs $G$ illustrated on \Cref{fig:lbG}, assume that $G$ is of the following form:
\begin{itemize}
\item $G$ has $m+1$ vertices $v_0,\ldots,v_m$ and $v_0$ is adjacent to $v_m$,
\item for each for each $i\in\{1,\ldots,m\}$, $G$ has a family  of vertex-disjoint paths $\mathcal{P}_i$ such that each path has at least one internal vertex. 
\end{itemize}
Suppose also that $c\colon V(G)\rightarrow U$ is a coloring function that colors the vertices of $G$ by colors from a set $U$. 
For each $i\in\{1,\ldots,m\}$, we define $\mathcal{F}_i=\{c(P)\colon P\in \mathcal{P}_i\}$, that is, $\mathcal{F}_i$ is the family of the sets of colors of the paths from $\mathcal{P}_i$. 
If $C$ is a cycle of $G$, then either (i) $C$ is formed by two paths $P,Q\in \mathcal{P}_i$ for some $i\in\{1,\ldots,m\}$ or (ii) $C$ contains the concatenation of $m$ paths $P_i\in \mathcal{P}_i$ for $i\in \{1,\ldots,m\}$. If we are looking for a cycle $C$ containing at most $k$ colors, we can use brute force to check whether there is such a cycle of type (i), because the number of such cycles is quadratic in the size of $G$. Suppose that this is not the case and we have  (ii). Then each family $\mathcal{F}_i$ is $k$-wide, and a cycle containing vertices of at most $k$ colors exists if and only if 
 there are $S_i\in\mathcal{F}_i$ for $i\in\{1,\ldots,m\}$ such that $|\bigcup_{i=1}^m S_i| \le k$. 

We first show that in contrast to the lower bound from \Cref{the:noquasip}, 
the \textsc{Wide Family Hitting} problem is fixed-parameter-tractable when parameterized by $k$.
We use the following lemma for it.

\begin{lemma}
\label{lem:setbound}
Let $X \subseteq U$ be a set and $\mathcal{F}$ a $k$-wide family of sets over $U$.
There are at most $2^{|X|}$ sets $A \in \mathcal{F}$ with $|A \cup X| \le k$ and $|A| \le |X|$.
\end{lemma}
\begin{proof}
Suppose there are sets $A,B \in \mathcal{F}$ with $A \cap X = B \cap X$, $|A \cup X| \le k$, $|B \cup X| \le k$, $|A| \le |X|$, and $|B| \le |X|$.
Then we have that
\[
|A \cup B| = |A| + |B| - |A \cap B| \le |A| + |B| - |A \cap X| \le |A| + |X| - |A \cap X| = |A \cup X| \le k,
\]
which would contradict the fact that $\mathcal{F}$ is $k$-wide.
Therefore, all sets $A \in \mathcal{F}$ with $|A \cup X| \le k$ and $|A| \le |X|$ have a different intersection with $X$, implying that there are at most $2^{|X|}$ of them.
\end{proof}

We will also use the following lemma in both of the algorithms of this section.

\begin{lemma}
\label{lem:onelight}
Let $X \subseteq U$ be a set with $|X| \le k$ and $\mathcal{F}$ a $k$-wide family of sets over $U$.
For any two sets $A,B \in \mathcal{F}$ it holds that $|X \cup A| - |X| + |X \cup B| - |X| > k - |X|$.
\end{lemma}
\begin{proof}
Note that $|X \cup A| + |X \cup B| - |X| \ge |A \cup B| > k$.
\end{proof}

Now we give our FPT algorithm.

\thmwfhfpt*
\begin{proof}
First, we guess the largest set $X = S_i$ selected to the solution.
Then we can remove all sets $A$ from the other families with $|A| > |X|$ or $|A \cup X| > k$.
Therefore, as $|X| \le k$, by \Cref{lem:setbound}, we can now assume that $|\mathcal{F}_i| \le 2^k$ for each $i$.

Then, we process the families $\mathcal{F}_i$ in an order from $i=1$ to $i=m$, accumulating a partial solution $X$ being the union of the selected sets so far.
Suppose that at index $i$, the family $\mathcal{F}_i$ contains a set $S_i$ with $S_i \subseteq X$.
Then, we can greedily include $S_i$ to the solution.
Otherwise, we branch on which set $S_i \in \mathcal{F}_i$ we include to the solution, which increases the size of our partial solution $X$.
As we can increase $X$ at most $k$ times and $|\mathcal{F}_i| \le 2^k$, this gives a $(2^k)^k N^{\Oh(1)} = 2^{\Oh(k^2)} N^{\Oh(1)}$ time algorithm.

To optimize the algorithm to $2^{\Oh(k \log k)} N^{\Oh(1)}$ time, we say that a set $A \in \mathcal{F}_i$ is \emph{light} with respect to the partial solution $X$ if $|X \cup A| - |X| \le (k - |X|)/2$.
In particular, a set $A$ is light if including it to $X$ decreases the remaining budget by at most half, while a set $A$ is \emph{heavy} if including it to $X$ decreases the remaining budget by more than a half.
By \Cref{lem:onelight}, $\mathcal{F}_i$ contains at most one light set.

In the branching, we can select a heavy set at most $\Oh(\log k)$ times, and otherwise we select a light set.
As there are $2^k$ options only when we select a heavy set and only one option when we select a light set, the time complexity becomes $2^k (2^k)^{\Oh(\log k)} N^{\Oh(1)} = 2^{\Oh(k \log k)} N^{\Oh(1)}$.
\end{proof}

Then, we give a polynomial-time algorithm when $|\mathcal{F}_i|$ is bounded.

\thmwfhbounddeg*
\begin{proof}
As in the proof of \Cref{the:wfh-fpt}, we again process the families $\mathcal{F}_i$ from $i=1$ to $i=m$, but this time instead of branching, we decide probabilistically which set to include in the solution.

At step $i$, let $X \subseteq U$ denote the accumulated partial solution so far (the union of the selected sets), and let $b = k - |X|$ be the remaining budget.
Again, as in \Cref{the:wfh-fpt}, we say that set $A \in \mathcal{F}_i$ is light if $|X \cup A| - |X| \le b/2$, i.e., including $A$ to the solution takes less than half of the remaining budget.
Otherwise a set $A \in \mathcal{F}_i$ is heavy.
By \Cref{lem:onelight}, there is at most one light set in $\mathcal{F}_i$.

First, if there is a set $A \in \mathcal{F}_i$ with $A \subseteq X$, we can greedily select the set $A$.
Otherwise, if $b = 0$ we must return that there is no solution, and if $b \ge 1$, our algorithm selects a set from $\mathcal{F}_i$ as follows.
If there is a light set $L \in \mathcal{F}_i$, let $c = |X \cup L| - |X| \ge 1$ be the cost of $L$.
Otherwise, we let $c = b/2$.
Note that in both cases $c \le b/2$.
By \Cref{lem:onelight}, the cost of any heavy set $H \in \mathcal{F}$ is $|X \cup H| - |X| > b-c$.
Our algorithm includes the light set $L$ to the solution with probability $\frac{b-c}{b}$ (if a light set exists), and any heavy set $H$ with probability $\frac{c}{b d}$.
Note that as $|\mathcal{F}_i| \le d$, these probabilities sum up to a number at most 1.

We claim that the probability that our algorithm finds a solution if one exists is at least 
\[
\frac{1}{2b} \cdot \left( \frac{1}{d} \right)^{1 + \log_2 b}.
\]
We prove this by induction on $b$.
The base case is that the set in $\mathcal{F}_i$ that belongs to the solution takes up all of the remaining budget, in particular, that a heavy set $H \in \mathcal{F}_i$ with $|X \cup H| - |X| = b$ is in the solution.
In this case, the algorithm is correct as long as it selects $H$ at this step, as the remaining steps will be deterministic.
As $c \ge 1/2$ and $b \ge 1$, the probability of correctness is
\[
\frac{c}{bd} \ge \frac{1}{2b} \cdot \frac{1}{d} \ge \frac{1}{2b}   \left( \frac{1}{d} \right)^{1 + \log_2 b},
\]
so the base case is satisfied.

Otherwise, a set from $\mathcal{F}_i$ that does not take all of the remaining budget belongs to the correct solution.
Suppose this set is light.
Now, by induction, the probability that the algorithm is correct is
\[
\frac{b-c}{b} \cdot \frac{1}{2 (b-c)} \cdot \left( \frac{1}{d} \right)^{1+\log_2 (b-c)} = \frac{1}{2b} \cdot \left( \frac{1}{d} \right)^{1+\log_2 (b-c)} \ge \frac{1}{2b} \cdot \left( \frac{1}{d} \right)^{1+\log_2 b},
\]
so the induction holds.

Then, suppose that the set from $\mathcal{F}_i$ that belongs to the correct solution is a heavy set $H \in \mathcal{F}_i$.
Note that in this case the remaining budget is $b - (|X \cup H| - |X|) < b - (b - c) < c < b/2$.
By induction, the algorithm is correct with probability
\[
\frac{c}{bd} \cdot \frac{1}{2(b - (|X \cup H| - |X|))} \cdot \left( \frac{1}{d} \right)^{1 + \log_2 (b - (|X \cup H| - |X|))} \ge \frac{c}{b} \cdot \frac{1}{2c} \cdot \left( \frac{1}{d} \right)^{2 + \log_2(b/2)} \ge \frac{1}{2b} \cdot \left(\frac{1}{d}\right)^{1 + \log_2 b}
\]

We have analyzed all cases in the induction, and therefore the algorithm is correct with probability $\frac{1}{2b} \cdot (1/d)^{1 + \log_2 b}$, and therefore (as initially $b=k$), repeating it $2k \cdot d^{\Oh(\log k)} = k^{\Oh(\log d)}$ times yields a correct result wth constant probability.
\end{proof}

\section{Conclusion}\label{sec:concl} 
We gave an $n^{\Oh(\log 1/\varepsilon)}$ time PTAS for the shortest monotone submodular cycle problem, and showed unconditional lower bounds establishing that this algorithm is optimal even in a very restricted setting, in particular even when the function is integer-valued, $\opt = \Oh(n)$, and the graph is planar and has bounded pathwidth.  

We leave several open questions. 
The main question about minimum cycles is the complexity of  \textsc{Hedge Minimum Cycle}. From what we know, there is no evidence against the existence of a polynomial-time algorithm. On the other hand, it also could be that our quasipolynomial-time algorithm for integer-valued monotone submodular functions is also optimal for  \textsc{Hedge Minimum Cycle}. 
This problem seems difficult, and therefore we believe it is worth exploring even some special cases of it.
In particular, we also ask if the \textsc{Wide Family Hitting} problem admits a polynomial-time algorithm.
While \Cref{the:wfh-fpt} shows that the special case of \textsc{Wide Family Hitting}  is fixed-parameter tractable, it remains a challenging question whether the \textsc{Hedge Minimum Cycle}  is fixed-parameter tractable in the general case. Of course, if the problem is in P this would resolve all these questions.

In the other direction, towards showing the hardness of  \textsc{Hedge Minimum Cycle}, we ask a purely combinatorial question which is a prerequisite for showing the hardness.
We say that a subset $S$ of the hedges is a \emph{minimal partial solution} if $|S| \le k$, and there is a pair of vertices $s,t$ so that $S$ induces a $(s,t)$-path, but no subset of $S$ induces an $(s,t)$-path.
We ask if there is a construction of a graph with hedges where the number of minimal partial solutions is superpolynomial.
Note that if the number of minimal partial solutions is polynomially bounded, then we can solve \textsc{Hedge Minimum Cycle} in polynomial time by a simple algorithm enumerating them.


\bibliographystyle{siam}

\begin{thebibliography}{10}

\bibitem{BroersmaLWZ05}
{\sc H.~Broersma, X.~Li, G.~J. Woeginger, and S.~Zhang}, {\em Paths and cycles
  in colored graphs}, Australas. {J} Comb., 31 (2005), pp.~299--312.

\bibitem{cunningham1985submodular}
{\sc W.~H. Cunningham}, {\em On submodular function minimization},
  Combinatorica, 5 (1985), pp.~185--192.

\bibitem{Diestel12}
{\sc R.~Diestel}, {\em Graph Theory, 4th Edition}, vol.~173 of Graduate texts
  in mathematics, Springer, 2012.

\bibitem{Dijkstra59}
{\sc E.~W. Dijkstra}, {\em A note on two problems in connexion with graphs},
  Numerische Mathematik, 1 (1959), pp.~269--271.

\bibitem{FredmanT84}
{\sc M.~L. Fredman and R.~E. Tarjan}, {\em Fibonacci heaps and their uses in
  improved network optimization algorithms}, in 25th Annual Symposium on
  Foundations of Computer Science, West Palm Beach, Florida, USA, 24-26 October
  1984, {IEEE} Computer Society, 1984, pp.~338--346.

\bibitem{ghaffari2017random}
{\sc M.~Ghaffari, D.~R. Karger, and D.~Panigrahi}, {\em Random contractions and
  sampling for hypergraph and hedge connectivity}, in Proceedings of the
  Twenty-Eighth Annual ACM-SIAM Symposium on Discrete Algorithms (SODA), SIAM,
  2017, pp.~1101--1114.

\bibitem{goel2009approximability}
{\sc G.~Goel, C.~Karande, P.~Tripathi, and L.~Wang}, {\em Approximability of
  combinatorial problems with multi-agent submodular cost functions}, in
  Proceedings of the 50th Annual IEEE Symposium on Foundations of Computer
  Science (FOCS), IEEE, 2009, pp.~755--764.

\bibitem{GoemansHIM09}
{\sc M.~X. Goemans, N.~J.~A. Harvey, S.~Iwata, and V.~S. Mirrokni}, {\em
  Approximating submodular functions everywhere}, in Proceedings of the
  Twentieth Annual {ACM-SIAM} Symposium on Discrete Algorithms (SODA), {SIAM},
  2009, pp.~535--544.

\bibitem{grotschel1981ellipsoid}
{\sc M.~Gr{\"o}tschel, L.~Lov{\'a}sz, and A.~Schrijver}, {\em The ellipsoid
  method and its consequences in combinatorial optimization}, Combinatorica, 1
  (1981), pp.~169--197.

\bibitem{iwata2001combinatorial}
{\sc S.~Iwata, L.~Fleischer, and S.~Fujishige}, {\em A combinatorial strongly
  polynomial algorithm for minimizing submodular functions}, Journal of the ACM
  (JACM), 48 (2001), pp.~761--777.

\bibitem{iwata2009submodular}
{\sc S.~Iwata and K.~Nagano}, {\em Submodular function minimization under
  covering constraints}, in Proceedings of the 50th Annual IEEE Symposium on
  Foundations of Computer Science (FOCS), IEEE, 2009, pp.~671--680.

\bibitem{iwata2009simple}
{\sc S.~Iwata and J.~B. Orlin}, {\em A simple combinatorial algorithm for
  submodular function minimization}, in Proceedings of the twentieth annual
  ACM-SIAM symposium on Discrete algorithms, SIAM, 2009, pp.~1230--1237.

\bibitem{JaffkeLMPS22}
{\sc L.~Jaffke, P.~T. Lima, T.~Masar{\'{\i}}k, M.~Pilipczuk, and U.~S. Souza},
  {\em A tight quasi-polynomial bound for global label min-cut}, CoRR,
  abs/2207.07426 (2022).

\bibitem{JaffkeLMPS23}
\leavevmode\vrule height 2pt depth -1.6pt width 23pt, {\em A tight
  quasi-polynomial bound for global label min-cut}, in Proceedings of the 23rd
  Annual {ACM-SIAM} Symposium on Discrete Algorithms (SODA), to appear, {SIAM},
  2023.

\bibitem{jegelka2009notes}
{\sc S.~Jegelka and J.~Bilmes}, {\em Notes on graph cuts with submodular edge
  weights}, in NIPS 2009 Workshop on Discrete Optimization in Machine Learning:
  Submodularity, Sparsity Polyhedra (DISCML), 2009, pp.~1--6.

\bibitem{schrijver2000combinatorial}
{\sc A.~Schrijver}, {\em A combinatorial algorithm minimizing submodular
  functions in strongly polynomial time}, Journal of Combinatorial Theory,
  Series B, 80 (2000), pp.~346--355.

\bibitem{svitkina2011submodular}
{\sc Z.~Svitkina and L.~Fleischer}, {\em Submodular approximation:
  Sampling-based algorithms and lower bounds}, SIAM Journal on Computing, 40
  (2011), pp.~1715--1737.

\bibitem{wolsey1982analysis}
{\sc L.~A. Wolsey}, {\em An analysis of the greedy algorithm for the submodular
  set covering problem}, Combinatorica, 2 (1982), pp.~385--393.

\bibitem{zhang2011approximation}
{\sc P.~Zhang, J.-Y. Cai, L.-Q. Tang, and W.-B. Zhao}, {\em Approximation and
  hardness results for label cut and related problems}, Journal of
  Combinatorial Optimization, 21 (2011), pp.~192--208.

\end{thebibliography}

\end{document}